\documentclass[aps,prd,reprint,superscriptaddress,nofootinbib,longbibliography,floatfix]{revtex4-2}

\usepackage[T1]{fontenc}
\usepackage[utf8]{inputenc}
\usepackage{lmodern}
\usepackage{microtype}
\usepackage{amsmath,amssymb,amsthm,mathtools,bm}
\usepackage{booktabs}
\usepackage{graphicx}
\usepackage{enumitem}
\usepackage{aliascnt}
\usepackage{hyperref}
\usepackage[nameinlink,noabbrev]{cleveref}

\hypersetup{
 hidelinks,
 pdftitle={Prelocality No-Go Theorems and Exactly Solvable Relational History Hamiltonians},
 pdfauthor={Maxim V. Churilov},
 pdfsubject={Permutation-covariant prelocal dynamics and exact relational history Hamiltonians},
 pdfkeywords={prelocality, permutation symmetry, history states, quantum graphity, spectral gap, relational dynamics}
}
\ifdefined\pdfinfoomitdate\pdfinfoomitdate=1\fi
\ifdefined\pdftrailerid\pdftrailerid{}\fi
\ifdefined\pdfsuppressptexinfo\pdfsuppressptexinfo=-1\fi
\setlist{itemsep=0.15em,topsep=0.35em,parsep=0pt}
\allowdisplaybreaks
\theoremstyle{plain}
\newtheorem{theorem}{Theorem}[section]
\newaliascnt{proposition}{theorem}
\newtheorem{proposition}[proposition]{Proposition}
\aliascntresetthe{proposition}
\newaliascnt{lemma}{theorem}
\newtheorem{lemma}[lemma]{Lemma}
\aliascntresetthe{lemma}
\newaliascnt{corollary}{theorem}
\newtheorem{corollary}[corollary]{Corollary}
\aliascntresetthe{corollary}
\newaliascnt{nogotheorem}{theorem}
\newtheorem{nogotheorem}[nogotheorem]{No-Go Theorem}
\aliascntresetthe{nogotheorem}
\newaliascnt{rigiditytheorem}{theorem}

\aliascntresetthe{rigiditytheorem}

\theoremstyle{definition}
\newaliascnt{definition}{theorem}

\aliascntresetthe{definition}
\newaliascnt{construction}{theorem}
\newtheorem{construction}[construction]{Construction}
\aliascntresetthe{construction}

\theoremstyle{remark}
\newaliascnt{remark}{theorem}
\newtheorem{remark}[remark]{Remark}
\aliascntresetthe{remark}
\newaliascnt{warning}{theorem}

\aliascntresetthe{warning}

\crefname{theorem}{Theorem}{Theorems}
\Crefname{theorem}{Theorem}{Theorems}
\crefname{nogotheorem}{No-Go Theorem}{No-Go Theorems}
\Crefname{nogotheorem}{No-Go Theorem}{No-Go Theorems}
\crefname{rigiditytheorem}{Rigidity Theorem}{Rigidity Theorems}
\Crefname{rigiditytheorem}{Rigidity Theorem}{Rigidity Theorems}
\crefname{proposition}{Proposition}{Propositions}
\Crefname{proposition}{Proposition}{Propositions}
\crefname{lemma}{Lemma}{Lemmas}
\Crefname{lemma}{Lemma}{Lemmas}
\crefname{corollary}{Corollary}{Corollaries}
\Crefname{corollary}{Corollary}{Corollaries}
\crefname{definition}{Definition}{Definitions}
\Crefname{definition}{Definition}{Definitions}
\crefname{construction}{Construction}{Constructions}
\Crefname{construction}{Construction}{Constructions}
\crefname{remark}{Remark}{Remarks}
\Crefname{remark}{Remark}{Remarks}
\crefname{warning}{Warning}{Warnings}
\Crefname{warning}{Warning}{Warnings}

\newcommand{\C}{\mathbb C}

\newcommand{\Z}{\mathbb Z}
\newcommand{\id}{\mathbf 1}
\newcommand{\Hh}{\mathcal H}
\newcommand{\D}{\mathcal D}
\newcommand{\T}{\mathcal T}

\newcommand{\Tr}{\operatorname{Tr}}
\newcommand{\Ker}{\operatorname{ker}}
\newcommand{\Ran}{\operatorname{Ran}}
\newcommand{\spec}{\operatorname{spec}}

\newcommand{\diam}{\operatorname{diam}}
\newcommand{\Aut}{\operatorname{Aut}}

\newcommand{\norm}[1]{\left\lVert #1\right\rVert}
\newcommand{\abs}[1]{\left\lvert #1\right\rvert}
\newcommand{\ket}[1]{\lvert #1\rangle}
\newcommand{\bra}[1]{\langle #1\rvert}
\newcommand{\ketbra}[2]{\lvert #1\rangle\!\langle #2\rvert}

\newcommand{\Hcyc}{H^{\mathrm{cyc}}_{N,L}}
\newcommand{\Hbool}{H^{\mathrm{B}}_N}
\newcommand{\EL}{\mathcal E_N}

\makeatletter
\AtBeginDocument{%
  \@ifpackageloaded{hyperref}{\hypersetup{hidelinks}}{}%
}
\makeatother

\begin{document}

\title{Prelocality No-Go Theorems and Exactly Solvable Relational History Hamiltonians}
\author{Maxim V. Churilov}
\email{churilovm1305@gmail.com}
\affiliation{Independent Researcher, Orenburg, Russia}
\date{October 8, 2026}

\begin{abstract}
Exact permutation covariance does not itself supply a sparse metric geometry.
For permutation-covariant, full-support $k$-body interaction orbits, the
incident interaction norm scales as
$\binom{N-1}{k-1}|J_N|\|h\|$; such orbits cannot have uniformly
finite range in metrics of diverging diameter.  A permutation-invariant
Hamiltonian diagonal in the graph basis has ground-state multiplicity at least
$N!/|\mathrm{Aut}(G)|$ whenever $G$ is a ground graph.  We give a
quantitative version for approximately symmetric energies, together with
fixed-radius connectivity and unitary-isometry obstructions.  The
trace-distance asymmetry of an arbitrary state also bounds its probability
of selecting one labeled graph.

We then construct two exactly solvable finite relational history Hamiltonians.
A cyclic clock with a reversible incidence update has a unique ground state;
gauge reduction yields a cycle Laplacian with a rank-one pin and the exact gap
$\Delta_{N,L}=E_* (1-\cos\theta_L)$, where
$\tan(L\theta_L/2)=1/\sin\theta_L$ and
$0<\theta_L<\pi/L$.  In particular,
$\Delta_{N,L}/E_*=\pi^2/(2L^2)-2\pi^2/L^3+O(L^{-4})$, independently of the
number of vertices.  In an alternative order-free architecture, one clock
qubit per unordered pair gives commuting four-local terms, $2^N$ ground
states and a size-independent gap $E_*(1-1/\sqrt2)$.  A fixed-weight
control penalty selects $\binom Nq$ ground states without reducing the gap.
These exact finite results expose clock and locality resource tradeoffs; they
do not establish an emergent geometric phase.
\end{abstract}

\maketitle

\section{Introduction}

A microscopic theory with no prescribed adjacency cannot use spatial locality
as an unexplained primitive.  If all microscopic labels are physically
indistinguishable, exact permutation covariance places all unordered pairs---and,
more generally, all $k$-subsets---in a single orbit.  An interaction supported
with equal strength on that orbit is naturally complete rather than sparse.
Rescaling its coefficient can control the norm, but rescaling does not create a
metric or a causal neighborhood.  Conversely, selecting a bounded-degree
support introduces either background data or a dynamical relation register.
This logical separation underlies quantum-graphity and related emergent-locality
programs, in which adjacency is promoted to a quantum degree of freedom and a
sparse phase is sought dynamically
\cite{Konopka2006,Konopka2008,Hamma2010,MarkopoulouSmolin2007}.  Explicit
model studies have also found large ground-state degeneracies in commonly used
graphity potentials \cite{SpectorSchwartz2018}; the orbit theorem below isolates
a symmetry-level mechanism that is independent of those particular potentials.

The first purpose of this paper is to make four prelocal obstructions precise.
The first is an orbit-growth theorem for covariant $k$-body interactions.  It
separates norm extensivity from metric locality: Kac-type scaling repairs the
former but not the latter.  The second is an orbit-degeneracy theorem: an
exactly permutation-invariant diagonal graph Hamiltonian cannot select one
nontrivial labeled sparse graph as a unique classical vacuum.  The third is a
local-information theorem for graph connectedness.  The multiset of all bounded-radius rooted neighborhoods does
not determine whether a finite graph is connected, so no energy or classifier
that factors only through that profile can be exact.  The fourth is a dynamical
isometry theorem.  Closed unitary evolution preserves every unitarily invariant
distance and therefore cannot contract an open family of states to a common
norm attractor.  These statements do not forbid emergent locality, statistical
connectedness, or equilibration of restricted observables.  They identify the
additional objects that must be constructed: a relational sector, a global
witness or asymptotic mechanism, and a weaker topology of dynamical approach.

The second purpose is constructive.  History-state Hamiltonians convert a
reversible update law into a stationary low-energy state
\cite{Feynman1985,KitaevShenVyalyi2002,Aharonov2007}.  Standard linear or cyclic
clocks lead to graph-Laplacian spectral problems; distributed clocks retain
more of the spacetime structure of a circuit
\cite{BreuckmannTerhal2014,BauschCrosson2018,CahaLandauNagaj2018}.  A companion
work treats arbitrary finite-dimensional unitary connections on a cycle, their
monodromy spectra, and operational clock changes \cite{ChurilovFiniteHistories2026}.
The present paper addresses a different problem: the update acts on dynamical
incidence registers and must coexist with exact microscopic permutation
covariance.  We analyze two complementary clocks.

The cyclic model has a single clock register and a unique relational history.
The earlier elementary estimate $\Delta=\Theta(L^{-2})$ can be replaced by a
complete solution.  The data penalty is a product-basis Hamming operator, and
the gauge-reduced Hamiltonian decomposes into rank-one perturbations of the
cycle Laplacian.  A scalar defect equation then gives the exact global gap and
its asymptotic expansion.  This also yields direct low-energy, thermal, and
perturbative certification bounds.

The Boolean-clock model addresses a different resource question.  Because all
pair updates commute, an order-free history is naturally indexed by subsets of
pair updates rather than by a chosen linear order.  One clock qubit per pair
realizes this subset space.  Every pair term is four-local in the tensor-factor
sense, the complete spectrum is obtained from a two-qubit calculation, and the
gap is constant.  The price is an exactly classified control-sector ground
manifold rather than an automatically selected unique permutation-invariant
state.  This constant-gap result does not contradict general criticality
theorems for universal circuit-to-Hamiltonian mappings
\cite{GonzalezCubitt2018}: the present update family is commuting, nonuniversal,
and retains an extensive family of control sectors.

The paper is organized as follows.  \Cref{sec:nogos} proves three classes of
prelocal obstruction, comprising four theorem-level no-go statements.  \Cref{sec:update} defines the reversible incidence algebra.
\Cref{sec:cyclic} solves the unique cyclic model, and
\Cref{sec:boolean} solves the order-free Boolean-clock model.
\Cref{sec:resources} compares their resources and claim boundaries.
Computational checks and adversarial tests are summarized in
\Cref{sec:verification}.

\section{Three prelocal obstructions}
\label{sec:nogos}

\subsection{Orbit growth versus metric locality}

Let $\mathfrak h$ be finite dimensional and
$\Hh_N=\mathfrak h^{\otimes N}$.  The symmetric group $S_N$ acts by the
canonical permutation representation $U_\pi$.  Fix $k\ge2$ and a
nonzero Hermitian $k$-site operator $h$ that is invariant under permutations of its
$k$ tensor factors and has \emph{minimal support equal to all $k$ sites}.
For every $k$-subset $X\subset[N]$, let $h_X$ denote its canonical
embedding on $X$ (therefore $\norm{h_X}=\norm h$), and set
\begin{equation}
 H_N^{(k)}=J_N\sum_{\substack{X\subset[N]\\ |X|=k}}h_X.
 \label{eq:kbody}
\end{equation}
The family is exactly permutation covariant and contains no preferred
$k$-subset.  Define the incident sum of interaction-term norms
\begin{equation}
 \Xi_N^{(k)}:=\max_{i\in[N]}
 \sum_{\substack{X\ni i\\ |X|=k}}\norm{J_Nh_X}.
 \label{eq:incident}
\end{equation}

\begin{nogotheorem}[Covariant orbit-growth and range obstruction]
\label{thm:orbit-trilemma}
For $N\ge k$, the family \eqref{eq:kbody} satisfies
\begin{equation}
 \Xi_N^{(k)}=\binom{N-1}{k-1}|J_N|\norm h.
 \label{eq:incident-exact}
\end{equation}
Consequently it cannot simultaneously possess:
\begin{enumerate}[label=(\roman*)]
\item exact $S_N$ covariance with no prescribed support hypergraph;
\item a volume-independent lower bound $|J_N|\ge J_0>0$;
\item a volume-uniform incident-norm bound;
\item finite interaction range in metrics $d_N$ whose diameters diverge.
\end{enumerate}
Kac scaling can enforce (iii) only by taking
$J_N=O(N^{1-k})$; for nonzero couplings it does not repair (iv).
The incident quantity in \eqref{eq:incident} is a sum of term norms,
not an asserted lower bound on the operator norm of the full Hamiltonian.
Retaining only $O(1)$ terms
incident on each site introduces either a background support hypergraph or a
dynamical support observable.
\end{nogotheorem}

\begin{proof}
The orbit of one $k$-subset under $S_N$ is the complete family
$\binom{[N]}k$.  Exactly $\binom{N-1}{k-1}$ members of that orbit contain a
fixed site, proving \eqref{eq:incident-exact}.  If $|J_N|\ge J_0$, the right
side grows as $N^{k-1}$.

Suppose instead that the interaction terms have metric range at most a
fixed $R$ and $J_N\ne0$.  By the minimal-support assumption,
each nonzero embedded term has support exactly $X$.  Given distinct
$i,j$ and $N\ge k$, some $k$-subset $X$ contains both.  Complete orbit support then implies $d_N(i,j)\le R$.  Hence
$\diam([N],d_N)\le R$, contradicting a diverging diameter.  Scaling $J_N$
does not alter support.  Finally, deleting almost all orbit terms requires
additional information specifying which subsets survive.  If fixed, it is
background structure; if represented by quantum degrees of freedom, locality
is sector dependent.
\end{proof}

Let \(\mathfrak G_N\) be the set of simple labeled graphs on \([N]\), with
orthonormal graph basis \(\{\ket G:G\in\mathfrak G_N\}\).  The permutation
representation acts by \(U_\pi\ket G=\ket{\pi G}\).

\begin{nogotheorem}[Symmetry-enforced graph-orbit degeneracy]
\label{thm:orbit-degeneracy}
Let
\begin{equation}
 H_{\mathrm{diag}}=\sum_{G\in\mathfrak G_N}E(G)\ketbra GG
 \label{eq:diag-graph-H}
\end{equation}
be diagonal in the graph basis and satisfy
\(U_\pi H_{\mathrm{diag}}U_\pi^\dagger=H_{\mathrm{diag}}\) for every
\(\pi\in S_N\).  If \(G\) is a ground graph, then every graph in its orbit is
a ground graph and
\begin{equation}
 \dim\Ker(H_{\mathrm{diag}}-E_0\id)
 \ge |S_N\!\cdot G|
 =\frac{N!}{|\Aut(G)|}.
 \label{eq:orbit-degeneracy}
\end{equation}
A unique graph-basis ground state is possible only when \(G\) is fixed by all
of \(S_N\); among simple graphs this leaves only the empty and complete graphs.
Consequently, a nonempty noncomplete sparse graph cannot be selected as a
unique classical vacuum by an exactly permutation-invariant diagonal energy.
\end{nogotheorem}

\begin{proof}
Permutation covariance and diagonality give \(E(\pi G)=E(G)\).  Hence the
entire orbit of a minimizing graph lies in the ground space.  Orbit--stabilizer
gives Eq.~\eqref{eq:orbit-degeneracy}.  If the orbit has one element, the edge
set is invariant under the transitive action of \(S_N\) on unordered pairs.
It is therefore either empty or the full pair set.
\end{proof}

\begin{corollary}[Robust orbit-gap obstruction]
\label{cor:robust-orbit-gap}
Let \(H_{\mathrm{diag}}\) be diagonal in the graph basis and define its
permutation-covariance defect
\begin{equation}
 \delta_{\mathrm{sym}}
 :=\max_{\pi\in S_N}
 \left\|U_\pi H_{\mathrm{diag}}U_\pi^\dagger-H_{\mathrm{diag}}\right\|.
 \label{eq:symmetry-defect}
\end{equation}
Suppose \(\ket G\) is a unique graph-basis ground state and let
\begin{equation}
 \gamma_G:=\min_{G'\ne G}\bigl[E(G')-E(G)\bigr]>0
 \label{eq:graph-gap}
\end{equation}
be its classical basis gap.  Then either \(G\) is empty or complete, or
\begin{equation}
 \gamma_G\le\delta_{\mathrm{sym}}.
 \label{eq:robust-orbit-gap}
\end{equation}
Thus a nontrivial labeled graph vacuum can be unique only at the price of a
symmetry defect at least as large as its protecting gap.
\end{corollary}

\begin{proof}
For every permutation \(\pi\),
\begin{align}
 |E(\pi G)-E(G)|
 &=\left|\bra G
 \bigl(U_\pi^\dagger H_{\mathrm{diag}}U_\pi-H_{\mathrm{diag}}\bigr)
 \ket G\right|\\
 &\le\delta_{\mathrm{sym}}.
\end{align}
If the orbit of \(G\) is nontrivial, choose \(\pi G\ne G\).  Uniqueness gives
\(E(\pi G)-E(G)\ge\gamma_G\), proving
Eq.~\eqref{eq:robust-orbit-gap}.  A singleton orbit is empty or complete by
\Cref{thm:orbit-degeneracy}.
\end{proof}

\begin{remark}[Coherent selection is a different mechanism]
The theorem concerns a unique \emph{graph-basis} vacuum of a diagonal
Hamiltonian.  An off-diagonal, permutation-invariant operator may select a
unique orbit superposition.  Such a state is not one labeled classical graph;
it requires coherent orbit mixing and leaves classical geometry to a further
measurement, decoherence, or symmetry-breaking analysis.
\end{remark}

\begin{corollary}[Necessary staged architecture]
\label{cor:staged}
Under the hypotheses of \Cref{thm:orbit-trilemma,thm:orbit-degeneracy}, a
permutation-covariant model with no hidden geometry must distinguish at least
three logically separate ingredients before claiming a physical finite-range
propagation law: a prelocal microscopic operator, selection of a sparse relational
sector, and a sectorwise quasillocality estimate such as a Lieb--Robinson bound
\cite{LiebRobinson1972}.
\end{corollary}

\begin{remark}
The theorem concerns complete interaction orbits.  It does not assert that
permutation symmetry forbids locality after spontaneous or relational sector
selection.  It forbids treating the selected metric support as already present
at the same logical level as the exact microscopic symmetry.
\end{remark}

\subsection{Bounded neighborhoods do not determine connectedness}

For an integer $r\ge0$, a finite graph $G$, and $v\in V(G)$, let $B_r(G,v)$ denote the rooted
radius-$r$ ball.  Define the rooted-neighborhood profile
\begin{equation}
 \mathsf P_r(G):=\bigl\{[B_r(G,v)]:v\in V(G)\bigr\}_{\mathrm{multi}},
 \label{eq:profile}
\end{equation}
where brackets denote rooted isomorphism type.  This profile contains strictly
more information than one additive local density: it determines every finite
collection of rooted radius-$r$ counts and every functional of their empirical
distribution, the basic data of local graph convergence
\cite{Lovasz2012}.

\begin{nogotheorem}[Local-information obstruction to connectedness]
\label{thm:connectivity}
Fix $r<\infty$.  No functional of $(|V(G)|,\mathsf P_r(G))$ can vanish on all
connected finite graphs and be strictly positive on all disconnected finite
graphs.  In particular, neither a fixed-radius additive energy nor an arbitrary
postprocessing of all radius-$r$ neighborhood counts can detect connectedness
exactly.
\end{nogotheorem}

\begin{proof}
Choose $n\ge\max\{3,2r+2\}$.  The connected cycle $C_{2n}$ and the disconnected graph
$C_n\sqcup C_n$ have the same number $2n$ of vertices.  Every rooted
radius-$r$ ball in either graph is the path on $2r+1$ vertices rooted at its
center.  Their profiles \eqref{eq:profile} are therefore identical, although
one graph is connected and the other is not.
\end{proof}

\begin{corollary}[Diagonal quantum penalties]
Let $H_r$ be diagonal in a graph basis and suppose its eigenvalue on $\ket G$
factors through $(|V(G)|,\mathsf P_r(G))$.  Then $H_r$ cannot have zero energy
on every connected graph and positive energy on every disconnected graph.
\end{corollary}

Exact finite-volume connectedness therefore requires at least one genuinely
global object: a spectral or Matrix--Tree penalty, an auxiliary spanning
witness, or a nonlocal constraint.  A weaker alternative is asymptotic
free-energy suppression, but then the theorem to prove is statistical rather
than exact.

\subsection{Closed unitary evolution has no common norm attractor}

A matrix norm $\lVert\!\lVert\cdot\rVert\!\rVert$ is unitarily invariant if
$\lVert\!\lVert UAV\rVert\!\rVert=\lVert\!\lVert A\rVert\!\rVert$ for all
unitaries $U,V$.  Schatten norms are the standard examples.

\begin{nogotheorem}[Unitary isometry obstruction]
\label{thm:unitary-attractor}
Let $\{U_s\}_{s\ge0}$ be any family of unitaries on a finite-dimensional
Hilbert space.  For every pair of states $\rho,\sigma$ and every unitarily
invariant norm,
\begin{equation}
 \lVert\!\lVert U_s\rho U_s^\dagger-U_s\sigma U_s^\dagger
 \rVert\!\rVert
 =\lVert\!\lVert\rho-\sigma\rVert\!\rVert.
 \label{eq:isometry}
\end{equation}
Hence two distinct initial states cannot converge under the same unitary
conjugation to one common norm limit.  In particular, no nonempty open set has
a point attractor in trace norm.
\end{nogotheorem}

\begin{proof}
Equation \eqref{eq:isometry} is unitary invariance.  If both evolved states
converged to the same $\rho_*$, the triangle inequality would force their
constant positive distance to converge to zero.
\end{proof}

\begin{corollary}[No attraction into an invariant sector]
If $P$ is a projection satisfying $[P,U_s]=0$ for all $s$, then
$\Tr[(\id-P)U_s\rho U_s^\dagger]=\Tr[(\id-P)\rho]$.  Nonzero weight outside
$\Ran P$ cannot disappear.
\end{corollary}

Reachability in a closed theory must therefore be formulated through a weaker
object: convergence of a fixed local algebra, temporal averages, typicality
inside conserved shells, decoherence after restriction, or conditioning on an
internal record.

\section{Reversible incidence algebra}
\label{sec:update}

Let $N\ge2$ and let
\begin{equation}
 \EL=\bigl\{\{i,j\}:1\le i<j\le N\bigr\},
 \qquad M:=|\EL|=\binom N2.
\end{equation}
Assign a control qubit $\mathcal C_i\simeq\C^2$ to each vertex and a relation
qubit $\mathcal R_e\simeq\C^2$ to each potential pair.  The data space is
\begin{equation}
 \D_N=\bigotimes_{i=1}^N\mathcal C_i\otimes
      \bigotimes_{e\in\EL}\mathcal R_e.
 \label{eq:data-space}
\end{equation}
Vertex labels are tensor coordinates; physical operators below are invariant
under their simultaneous permutation.

For $e=\{i,j\}$, set
\begin{equation}
 Q_e=\id-\ketbra{00}{00}_{\mathcal C_i\mathcal C_j},
 \qquad
 W_e=(\id-Q_e)\otimes\id+Q_e\otimes X_e,
 \label{eq:local-update}
\end{equation}
where $X_e$ flips $\mathcal R_e$.  Thus the relation bit is toggled if at least
one endpoint is active.

\begin{lemma}[Commuting involutive updates]
\label{lem:update-algebra}
For all $e,f\in\EL$,
\begin{equation}
 W_e=W_e^\dagger,
 \qquad W_e^2=\id,
 \qquad [W_e,W_f]=0.
 \label{eq:update-algebra}
\end{equation}
Consequently
\begin{equation}
 W_N:=\prod_{e\in\EL}W_e
 \label{eq:global-update}
\end{equation}
is an order-independent Hermitian involution.
\end{lemma}

\begin{proof}
All $Q_e$ are diagonal in the common computational basis of the controls.
The relation flips act on distinct pair registers.  The two summands in
\eqref{eq:local-update} have orthogonal control support, so $W_e^2=\id$ and
$W_e=W_e^\dagger$.  Common control diagonality and disjoint relation support
give pairwise commutation.
\end{proof}

The same local update algebra supports two inequivalent histories.  A cyclic
clock applies the complete involution $W_N$ at each step and permits a unique
input-selected history.  A Boolean clock records independently which commuting
pair updates have occurred and avoids choosing an order on $\EL$.

\section{Unique cyclic history: complete spectrum and exact gap}
\label{sec:cyclic}

\subsection{Hamiltonian and gauge reduction}

Let $L\ge4$ be even and
$\T_L=\operatorname{span}\{\ket t:t\in\Z_L\}$.  For each clock edge define
\begin{equation}
 A_t\psi=\ket t\otimes\psi-\ket{t+1}\otimes W_N\psi,
 \qquad
 P_t=\frac12A_tA_t^\dagger.
 \label{eq:history-projector}
\end{equation}
Because $A_t^\dagger A_t=2\id$, every $P_t$ is an orthogonal projector.  The
zero condition for $\Psi=\sum_t\ket t\otimes\psi_t$ is
$\psi_{t+1}=W_N\psi_t$.

Choose
\begin{equation}
 \ket{\Omega_N}=\ket+_{\mathcal C}^{\otimes N}
 \otimes\ket0_{\mathcal R}^{\otimes M}
 \label{eq:omega}
\end{equation}
and the product-basis penalty
\begin{equation}
 K_N=\sum_{i=1}^N(\id-\ketbra{+}{+}_i)
 +\sum_{e\in\EL}\ketbra{1}{1}_e.
 \label{eq:K}
\end{equation}
Its kernel is $\C\ket{\Omega_N}$.

\begin{construction}[Cyclic incidence-history Hamiltonian]
\label{con:cyclic}
For one scale $E_*>0$, define
\begin{equation}
 \boxed{\Hcyc=E_*\left[
 \sum_{t\in\Z_L}P_t+\ketbra{0}{0}_{\T_L}\otimes K_N
 \right].}
 \label{eq:Hcyc}
\end{equation}
\end{construction}

The operator is bounded, self-adjoint, nonnegative, and exactly
permutation covariant.  Since $W_N^2=\id$ and $L$ is even, the history closes.
Define the block-diagonal gauge
\begin{equation}
 G_{N,L}=\sum_{t\in\Z_L}\ketbra{t}{t}\otimes W_N^t
 \label{eq:gauge}
\end{equation}
and the normalized cycle Laplacian
\begin{equation}
 L_L=\frac12\sum_{t\in\Z_L}
 (\ket t-\ket{t+1})(\bra t-\bra{t+1}).
 \label{eq:cycle-lap}
\end{equation}

\begin{theorem}[Exact gauge reduction and unique ground state]
\label{thm:gauge-ground}
The cyclic Hamiltonian obeys
\begin{equation}
 G_{N,L}^\dagger\Hcyc G_{N,L}
 =E_*\left[L_L\otimes\id+\ketbra{0}{0}\otimes K_N\right].
 \label{eq:gauge-reduction}
\end{equation}
It has the unique normalized zero-energy state
\begin{equation}
 \ket{\Psi_{N,L}}=\frac1{\sqrt L}
 \sum_{t=0}^{L-1}\ket t\otimes W_N^t\ket{\Omega_N}.
 \label{eq:cyclic-ground}
\end{equation}
\end{theorem}

\begin{proof}
On the $t$th clock fiber, $G_{N,L}$ multiplies the data vector by $W_N^t$.
The update across $t\to t+1$ cancels because
$W_N^{t+1}=W_NW_N^t$, turning every twisted difference into the ordinary
cycle difference.  The pin lies at $t=0$ and is unchanged, giving
\eqref{eq:gauge-reduction}.

A zero-energy vector is constant in clock after the gauge transformation and
its data vector lies in $\Ker K_N=\C\Omega_N$.  Undoing the gauge gives
\eqref{eq:cyclic-ground}.  Therefore the kernel is one dimensional.
\end{proof}

\subsection{Complete block spectrum}

Set
\begin{equation}
 R:=N+M=\frac{N(N+1)}2.
 \label{eq:R}
\end{equation}
The summands of $K_N$ are independent one-qubit projectors in a product basis.
Hence $K_N$ has eigenvalue $m\in\{0,\ldots,R\}$ with multiplicity
$\binom Rm$.  Define the clock defect operator
\begin{equation}
 A_{L,m}:=L_L+m\ketbra{0}{0}.
 \label{eq:ALm}
\end{equation}

\begin{theorem}[Complete spectral reduction]
\label{thm:complete-spectrum}
The spectrum of the cyclic Hamiltonian, including multiplicity, is
\begin{equation}
 \spec(\Hcyc/E_*)=
 \biguplus_{m=0}^{R}\binom Rm\,\spec(A_{L,m}).
 \label{eq:full-spectrum}
\end{equation}
Let
\begin{equation}
 \lambda_k=1-\cos\frac{2\pi k}{L},
 \qquad k=0,\ldots,L-1.
 \label{eq:lambda-k}
\end{equation}
For $z\notin\{\lambda_k\}$, the characteristic polynomial of $A_{L,m}$ is
\begin{equation}
 \det(A_{L,m}-z\id)=
 \prod_{k=0}^{L-1}(\lambda_k-z)
 \left[1+\frac mL\sum_{k=0}^{L-1}\frac1{\lambda_k-z}\right],
 \label{eq:secular}
\end{equation}
with the right side understood by polynomial continuation at the apparent
poles.  Equation \eqref{eq:secular} determines every eigenvalue of every data
sector.
\end{theorem}

\begin{proof}
Diagonalize $K_N$ in its product eigenbasis and use
\eqref{eq:gauge-reduction}; each data eigenvector of weight $m$ carries the
clock operator $A_{L,m}$, with multiplicity $\binom Rm$.  The Fourier basis
diagonalizes $L_L$ with eigenvalues \eqref{eq:lambda-k}, and
$|\langle0|k\rangle|^2=1/L$.  The matrix determinant lemma gives
\eqref{eq:secular}.
\end{proof}

The reduction is stronger than a gap estimate: all $N$-dependence is confined
to binomial multiplicities.  The numerical values of the clock eigenvalues
depend only on $L$ and the integer pin strength $m$.

\subsection{Exact global gap}

\begin{theorem}[Exact $N$-independent gap]
\label{thm:exact-gap}
Let $L\ge4$ be even.  There is a unique
$\theta_L\in(0,\pi/L)$ satisfying
\begin{equation}
 \tan\frac{L\theta_L}{2}=\frac1{\sin\theta_L}.
 \label{eq:theta-equation}
\end{equation}
The first positive eigenvalue of \eqref{eq:Hcyc} is exactly
\begin{equation}
 \boxed{\Delta_{N,L}=E_*\bigl(1-\cos\theta_L\bigr),}
 \label{eq:exact-gap}
\end{equation}
for every $N\ge2$.  It obeys the sharp elementary bracket
\begin{equation}
 E_*\left(1-\cos\frac\pi{L+2}\right)
 <\Delta_{N,L}<
 E_*\left(1-\cos\frac\pi L\right)
 \label{eq:gap-bracket}
\end{equation}
and the expansion
\begin{equation}
 \Delta_{N,L}=E_*\left[
 \frac{\pi^2}{2L^2}-\frac{2\pi^2}{L^3}+O(L^{-4})
 \right].
 \label{eq:gap-asymptotic}
\end{equation}
\end{theorem}

\begin{proof}
The $m=0$ block has the constant zero mode and first excitation
$1-\cos(2\pi/L)$.  For $m\ge1$, the lowest eigenvalue of $A_{L,m}$ is
monotone nondecreasing in $m$ by the min--max principle.  It is therefore
enough to solve the ground eigenvalue of $A_{L,1}$.

Write $\mu=1-\cos\theta$ with $0<\theta<\pi/L$ and choose the positive,
reflection-symmetric vector
\begin{equation}
 f_t=\cos\left[\left(\frac L2-t\right)\theta\right],
 \qquad t\in\Z_L.
 \label{eq:defect-vector}
\end{equation}
Away from $t=0$, the cycle recurrence gives
$L_Lf=(1-\cos\theta)f$.  At the defect site, the eigenvalue equation is
\begin{equation}
 (1+\cos\theta)f_0-f_1=0.
\end{equation}
Substitution of \eqref{eq:defect-vector} reduces this condition to
\eqref{eq:theta-equation}.  On $(0,\pi/L)$ the left side of
\eqref{eq:theta-equation} increases from zero to infinity, while
$1/\sin\theta$ decreases, so the root is unique.  The vector is strictly
positive; irreducibility of the defect Laplacian identifies it as the ground
state of $A_{L,1}$.

At $\theta_0=\pi/(L+2)$,
$\tan(L\theta_0/2)=\cot\theta_0<1/\sin\theta_0$, while the tangent diverges as
$\theta\uparrow\pi/L$.  Thus
$\pi/(L+2)<\theta_L<\pi/L$, proving \eqref{eq:gap-bracket}.  This defect
eigenvalue lies below $1-\cos(\pi/L)$ and hence below the first positive
$m=0$ eigenvalue.  It is therefore the global gap.

Finally, set $\delta_L=\pi-L\theta_L$.  Equation
\eqref{eq:theta-equation} becomes
$\cot(\delta_L/2)=1/\sin[(\pi-\delta_L)/L]$.  Taylor expansion gives
$\delta_L=2\pi/L+O(L^{-2})$, hence
$\theta_L=\pi/L-2\pi/L^2+O(L^{-3})$.  Expanding $1-\cos\theta_L$ yields
\eqref{eq:gap-asymptotic}.
\end{proof}

\begin{figure}[t]
\centering
\includegraphics[width=\columnwidth]{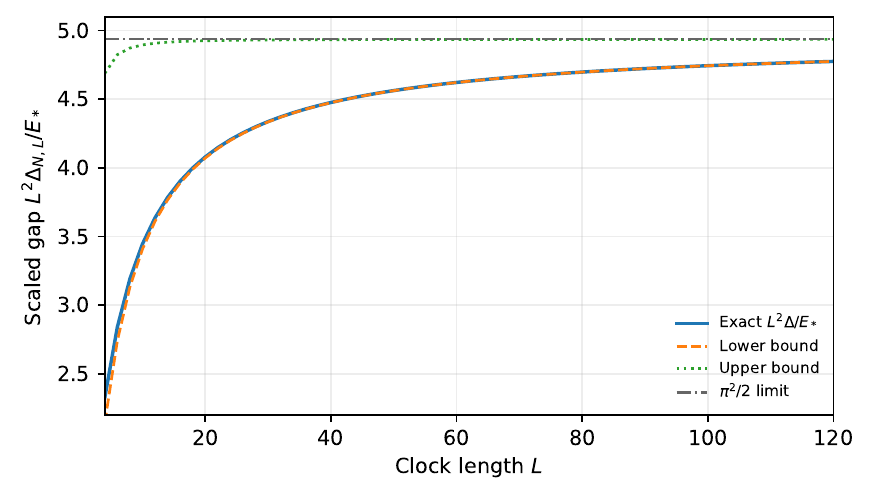}
\caption{Exact scaled cyclic gap from \cref{eq:theta-equation}, together with
the two analytic bounds in \cref{eq:gap-bracket}.  The scaled gap tends to
$\pi^2/2$; the complete result is independent of the number of vertices.}
\label{fig:gap}
\end{figure}

\subsection{Low-energy, thermal, and perturbative certification}

Let $P_{N,L}=\ketbra{\Psi_{N,L}}{\Psi_{N,L}}$ and let
$D_{N,L}=L2^R$ be the Hilbert-space dimension.

\begin{theorem}[History certification]
\label{thm:certification}
For any state $\rho$,
\begin{align}
 1-\Tr(P_{N,L}\rho)
 &\le\frac{\Tr(\rho\Hcyc)}{\Delta_{N,L}},
 \label{eq:energy-fidelity}\\
 \frac12\norm{\rho-P_{N,L}}_1
 &\le\sqrt{\frac{\Tr(\rho\Hcyc)}{\Delta_{N,L}}}.
 \label{eq:trace-cert}
\end{align}
For the Gibbs state
$\rho_\beta=e^{-\beta\Hcyc}/\Tr(e^{-\beta\Hcyc})$,
\begin{equation}
 1-\Tr(P_{N,L}\rho_\beta)
 \le(D_{N,L}-1)e^{-\beta\Delta_{N,L}}.
 \label{eq:gibbs-cert}
\end{equation}
\end{theorem}

\begin{proof}
The spectral theorem gives
$\Hcyc\ge\Delta_{N,L}(\id-P_{N,L})$, proving
\eqref{eq:energy-fidelity}.  The trace-distance bound follows from the
Fuchs--van de Graaf inequality for a pure reference state
\cite{NielsenChuang2010}.  For the Gibbs state, every excited Boltzmann weight
is at most $e^{-\beta\Delta_{N,L}}$, and there are at most
$D_{N,L}-1$ excited levels counting multiplicity.
\end{proof}

\begin{theorem}[Finite perturbative stability]
\label{thm:stability}
Let $V=V^\dagger$ and
$\widetilde H=\Hcyc+V$.  If
$\varepsilon:=\norm V<\Delta_{N,L}/2$, then the ground eigenvalue of
$\widetilde H$ is simple, its gap satisfies
\begin{equation}
 \widetilde\Delta\ge\Delta_{N,L}-2\varepsilon,
 \label{eq:perturbed-gap}
\end{equation}
and its rank-one ground projector $\widetilde P$ obeys
\begin{equation}
 \norm{\widetilde P-P_{N,L}}
 \le\min\left\{1,\frac{\varepsilon}{\Delta_{N,L}-2\varepsilon}\right\}.
 \label{eq:projector-bound}
\end{equation}
In particular, the admissible full-operator-norm window scales as $L^{-2}$.
\end{theorem}

\begin{proof}
Weyl inequalities shift each ordered eigenvalue by at most $\varepsilon$,
which gives \eqref{eq:perturbed-gap} and simplicity.  For the projector bound, write $P=P_{N,L}$, $Q=\id-P$, and let
$\ket{\varphi}$ be a normalized perturbed ground state of energy
$\widetilde E_0$.  Weyl's bound gives
$|\widetilde E_0|\le\varepsilon$.
The compression
$Q(\widetilde H-\widetilde E_0)Q$ is bounded below by
$(\Delta_{N,L}-2\varepsilon)Q$ on $\Ran Q$.  Projecting the
eigenvalue equation onto $\Ran Q$ therefore gives
$\norm{Q\varphi}\le\varepsilon/(\Delta_{N,L}-2\varepsilon)$.
For rank-one orthogonal projectors
$\norm{\widetilde P-P}=\norm{Q\varphi}\le1$, yielding
\eqref{eq:projector-bound}; cf. standard matrix perturbation theory
\cite{Bhatia1997}.
\end{proof}

The theorem is finite-system stability, not a thermodynamic phase theorem.
Because the gap closes as $L^{-2}$, an extensive perturbation cannot be
controlled by its full operator norm uniformly in $L$.  A genuine phase result
would require an interaction norm, local spectral flow, or an independently
isolated sector.

\section{Order-free Boolean histories with a constant gap}
\label{sec:boolean}

The cyclic construction packages the complete update $W_N$ into one clock
transition.  Since the factors $W_e$ commute, there is another natural
history: record the subset of pair updates already applied, without choosing an
ordering of $\EL$.

Assign a clock qubit $\mathcal T_e\simeq\C^2$ to every $e\in\EL$ and let
\begin{equation}
 \T_N^{\mathrm B}=\bigotimes_{e\in\EL}\mathcal T_e.
 \label{eq:boolean-clock}
\end{equation}
Denote by $X_e^{\T}$ the Pauli $X$ on $\mathcal T_e$.  Define
\begin{equation}
 h_e=\frac12\left(\id-X_e^{\T}\otimes W_e\right)
 +\ketbra{0}{0}_{\mathcal T_e}\otimes\ketbra{1}{1}_{\mathcal R_e}.
 \label{eq:local-boolean-term}
\end{equation}
The first summand is a propagation projector because
$X_e^{\T}\otimes W_e$ is a Hermitian involution.  The second fixes the
relation bit before its local clock fires.

\begin{construction}[Boolean incidence-history Hamiltonian]
\label{con:boolean}
Define
\begin{equation}
 \boxed{\Hbool=E_*\sum_{e\in\EL}h_e.}
 \label{eq:Hboolean}
\end{equation}
Every $h_e$ acts on one clock qubit, one relation qubit, and the two endpoint
controls.  Thus the interaction arity is at most four, independent of $N$.
No order on the pair set appears.
\end{construction}

For a computational control string $c\in\{0,1\}^N$, set
\begin{equation}
 q_e(c)=1-\delta_{c_i,0}\delta_{c_j,0}
 \qquad(e=\{i,j\}).
 \label{eq:qe}
\end{equation}
On this control block, $W_e$ is $\id$ if $q_e(c)=0$ and $X$ if
$q_e(c)=1$.  Introduce
\begin{equation}
 \ket{\phi_0}=\ket+_{\mathcal T}\ket0_{\mathcal R},
 \qquad
 \ket{\phi_1}=\frac{\ket{00}+\ket{11}}{\sqrt2}.
 \label{eq:local-ground-vectors}
\end{equation}

\begin{theorem}[Exact Boolean-history solution]
\label{thm:boolean-solution}
The terms $h_e$ commute for distinct pairs.  The complete ground space of
\eqref{eq:Hboolean} has orthonormal basis
\begin{equation}
 \ket{\Phi_c}=\ket c_{\mathcal C}\otimes
 \bigotimes_{e\in\EL}\ket{\phi_{q_e(c)}}_{\mathcal T_e\mathcal R_e},
 \qquad c\in\{0,1\}^N,
 \label{eq:boolean-ground}
\end{equation}
so $\dim\Ker\Hbool=2^N$.  Its full spectrum is the additive multiset generated
by
\begin{equation}
 \left\{0,\ 1-\frac1{\sqrt2},\ 1,\ 1+\frac1{\sqrt2}\right\}E_*
 \label{eq:local-spectrum}
\end{equation}
for each of the $M$ pairs, repeated over the $2^N$ control blocks.  In
particular,
\begin{equation}
 \boxed{\Delta_N^{\mathrm B}=E_*\left(1-\frac1{\sqrt2}\right)}
 \label{eq:boolean-gap}
\end{equation}
for every $N\ge2$.
\end{theorem}

\begin{proof}
The control projectors $Q_e$ are simultaneously diagonal.  Distinct terms
$h_e$ and $h_f$ act on disjoint clock--relation pairs, so they commute even
when their endpoint sets overlap.

Fix a control string $c$.  The local term is either
\begin{align}
 h_0&=\frac12(\id-X\otimes\id)+\ketbra{0}{0}\otimes\ketbra{1}{1},
 \label{eq:h0}\\
 h_1&=\frac12(\id-X\otimes X)+\ketbra{0}{0}\otimes\ketbra{1}{1}.
 \label{eq:h1}
\end{align}
A CNOT from the clock qubit to the relation qubit conjugates $h_0$ to $h_1$.
In the relation-$0$ block, $h_0$ has eigenvalues $0,1$; in the relation-$1$
block it has matrix
\begin{equation}
 \begin{pmatrix}3/2&-1/2\\-1/2&1/2\end{pmatrix}
\end{equation}
with eigenvalues $1\pm1/\sqrt2$.  The unique local ground vectors are exactly
\eqref{eq:local-ground-vectors}.  Tensoring the unique local ground vector for
each pair yields \eqref{eq:boolean-ground}, and no additional zero modes are
possible.  Additivity of commuting pair blocks gives the full spectrum and
\eqref{eq:boolean-gap}.
\end{proof}

For fixed $c$, the ground vector has the explicit order-free history form
\begin{equation}
 \ket{\Phi_c}=2^{-M/2}\sum_{S\subseteq\EL}
 \ket S_{\mathcal T}\otimes
 \left(\prod_{e\in S}W_e\right)
 \ket c_{\mathcal C}\ket{0^M}_{\mathcal R}.
 \label{eq:subset-history}
\end{equation}
Because the $W_e$ commute, only the subset $S$ matters; no linear extension or
preferred pair ordering is hidden in the clock.

\subsection{Fixed-weight sectors and invariant histories}

Let $n_i=\ketbra{1}{1}_i$ and, for $q\in\{0,\ldots,N\}$, define
\begin{equation}
 C_q=\left(\sum_{i=1}^N n_i-q\id\right)^2,
 \qquad
 H_{N,q}^{\mathrm B}=\Hbool+E_*C_q.
 \label{eq:q-sector}
\end{equation}
The square expands into one- and two-control terms on the complete pair orbit,
so it is prelocal but exactly permutation invariant.

\begin{corollary}[Exact control-sector selection]
\label{cor:q-sector}
The ground space of $H_{N,q}^{\mathrm B}$ is
\begin{equation}
 \operatorname{span}\{\Phi_c:|c|=q\},
 \qquad
 \dim\Ker H_{N,q}^{\mathrm B}=\binom Nq,
 \label{eq:q-ground}
\end{equation}
and its gap remains $E_*(1-2^{-1/2})$.  The normalized vector
\begin{equation}
 \ket{\Phi_{N,q}^{\mathrm{sym}}}=\binom Nq^{-1/2}
 \sum_{|c|=q}\ket{\Phi_c}
 \label{eq:symmetric-history}
\end{equation}
spans the one-dimensional $S_N$-invariant subspace of this orbit sector.  For $q=0$ or
$q=N$, the ground state is unique.
\end{corollary}

\begin{proof}
The operator $C_q$ commutes with $\Hbool$, vanishes exactly on control strings
of weight $q$, and has first positive eigenvalue one.  Since the Boolean gap is
smaller than one, adding $C_q$ does not change the gap.  Transitivity of $S_N$
on the weight-$q$ strings makes the invariant subspace one-dimensional.
\end{proof}

The symmetric vector \eqref{eq:symmetric-history} is selected kinematically
inside the sector but is not the unique ground state unless an additional
orbit-mixing or global symmetry selector is supplied.  This is a deliberate
claim boundary: exact microscopic covariance does not by itself produce a
unique symmetry-invariant pure state.

\subsection{Sharp clock-memory lower bound}

\begin{proposition}[Orthogonal subset-clock cost]
\label{prop:clock-memory}
Suppose a sharp order-free clock assigns mutually orthogonal clock events to
all subsets $S\subseteq\EL$ of $M$ independently recordable commuting updates.
Then
\begin{equation}
 \dim\T\ge2^M,
 \qquad
 \log_2\dim\T\ge M.
 \label{eq:clock-lower-bound}
\end{equation}
The Boolean clock \eqref{eq:boolean-clock} saturates both bounds.
\end{proposition}

\begin{proof}
There are $2^M$ subsets, and a Hilbert space containing $2^M$ mutually
orthogonal nonzero vectors has dimension at least $2^M$.  The tensor product of
$M$ clock qubits has exactly that dimension.
\end{proof}

The proposition does not rule out compressed nonorthogonal, probabilistic, or
coarse-grained clocks.  It is a sharp statement for exact distinguishability of
all subset histories.

\section{Symmetry, resources, and claim boundary}
\label{sec:resources}

A permutation $\pi\in S_N$ acts simultaneously on controls, relation registers,
and---in the Boolean model---pair clocks by
$e=\{i,j\}\mapsto\pi e=\{\pi(i),\pi(j)\}$.  It acts trivially on the single
cyclic clock.

\begin{theorem}[Exact permutation covariance]
\label{thm:permutation}
For every $\pi\in S_N$,
\begin{align}
 [U_\pi,\Hcyc]&=0, &
 [U_\pi,\Hbool]&=0,\label{eq:permutation}\\
 U_\pi\ket{\Psi_{N,L}}&=\ket{\Psi_{N,L}}.&&
\end{align}
Moreover $U_\pi\ket{\Phi_c}=\ket{\Phi_{\pi c}}$ and
$U_\pi\ket{\Phi_{N,q}^{\mathrm{sym}}}=\ket{\Phi_{N,q}^{\mathrm{sym}}}$.
\end{theorem}

\begin{proof}
Conjugation permutes the family $\{Q_e,W_e\}_{e\in\EL}$ and the corresponding
clock--relation factors.  Every Hamiltonian is a sum or product over the full
pair orbit.  The initial cyclic state is symmetric within each register type,
and the Boolean ground vectors transform by relabeling the control string.
\end{proof}

\begin{center}
\refstepcounter{table}\label{tab:comparison}
\parbox{0.98\columnwidth}{\small\textbf{TABLE \thetable.} Exact resource comparison. ``Arity'' refers to the displayed tensor factorization, not to emergent spatial range.}
\vspace{3pt}
\begin{ruledtabular}
\begin{tabular}{@{}lcc@{}}
Resource & Cyclic model & Boolean model\\
\hline
Clock event space & $L$ & $2^M$\\
Clock qubits & $\lceil\log_2L\rceil$ & $M$\\
Pair-update ordering & cyclic global step & none\\
Propagation arity & global $W_N$ & at most $4$\\
Ground dimension & $1$ & $2^N$\\
With weight-$q$ selector & $1$ & $\binom Nq$\\
Exact gap & Eq.~\eqref{eq:exact-gap} & $E_*(1-2^{-1/2})$\\
Large-size behavior & $\Theta(L^{-2})$ & constant in $N$\\
\end{tabular}
\end{ruledtabular}
\end{center}

The two architectures solve different finite problems.  The cyclic model uses
one compact clock and a global commuting update to obtain a unique history.
The Boolean model resolves every pair update locally and removes the order, but
retains control-sector degeneracy and an exponentially large sharp event
space.  Neither architecture evades \cref{thm:orbit-trilemma}: the cyclic
propagation is explicitly prelocal, while the Boolean control qubits have
interaction degree $N-1$ through the complete pair orbit.  Bounded tensor arity
is not the same as bounded range in an emergent metric.

The following statements are therefore outside the paper:
\begin{enumerate}[label=(\roman*)]
\item selection of a sparse connected graph phase;
\item a volume-uniform Lieb--Robinson cone in a selected geometry;
\item a thermodynamic phase with stable local correlators;
\item chiral matter, a massless spin-two sector, or Einstein dynamics;
\item dynamical attraction of generic initial states to either exact history.
\end{enumerate}
The constructions are finite prelocal operator models.  They make the
remaining phase-selection and locality obligations explicit rather than
absorbing them into terminology.  The orbit--stabilizer, rooted-cycle,
and unitary-isometry facts used above are elementary ingredients; the
constructive statements concern the specified permutation-covariant
Hamiltonians and their exact finite spectra, not a claimed universal
locality or quantum-computation theorem.

\section{State-level cost of selecting one labeling}
\label{sec:state-concentration}

The orbit obstruction is not confined to diagonal Hamiltonians or classical
mixtures.  It also constrains arbitrary coherent and mixed states.
Let \(\mathcal O_G\) be the orthonormal graph-basis orbit of a labeled graph
\(G\), with
\[
 M_G=\abs{\mathcal O_G}=\frac{N!}{\abs{\Aut(G)}}.
\]
For a density operator \(\rho\), define its permutation twirl and worst-case
state covariance defect by
\begin{align}
 \mathcal G(\rho)&=\frac1{N!}\sum_{\pi\in S_N}U_\pi\rho U_\pi^\dagger,
 \\
 \delta_{\rm st}(\rho)&=
 \max_{\pi\in S_N}\frac12\norm{\rho-U_\pi\rho U_\pi^\dagger}_1.
\end{align}

\begin{theorem}[Coherent labeling-concentration bound]
\label{thm:state-concentration}
For every density operator \(\rho\) and every labeled graph basis vector
\(\ket G\),
\begin{equation}
 \boxed{
 \bra G\rho\ket G
 \le \frac1{M_G}+\delta_{\rm st}(\rho).}
 \label{eq:state-concentration}
\end{equation}
Consequently, concentration
\(\bra G\rho\ket G\ge1-\varepsilon\) requires
\begin{equation}
 \delta_{\rm st}(\rho)
 \ge1-\varepsilon-\frac1{M_G}.
 \label{eq:state-asymmetry-cost}
\end{equation}
\end{theorem}

\begin{proof}
Convexity of trace distance gives
\[
 \frac12\norm{\rho-\mathcal G(\rho)}_1
 \le\frac1{N!}\sum_{\pi}\frac12
 \norm{\rho-U_\pi\rho U_\pi^\dagger}_1
 \le\delta_{\rm st}(\rho).
\]
Because the twirled state is constant on the orbit diagonal,
\[
 \bra G\mathcal G(\rho)\ket G
 =\frac1{M_G}\sum_{H\in\mathcal O_G}\bra H\rho\ket H
 \le\frac1{M_G}.
\]
The variational characterization of trace distance bounds the change of the
projector probability \(\ketbra GG\), proving
\eqref{eq:state-concentration}; rearrangement gives
\eqref{eq:state-asymmetry-cost}.
\end{proof}

The same conclusion admits an entropic certificate.  Write
\[
 A_{\rm rel}(\rho)=
 D\!\left(\rho\middle\|\mathcal G(\rho)\right).
\]

\begin{corollary}[Relative-entropy labeling cost]
\label{cor:relative-labeling-cost}
Suppose $M_G>1$ and let
\(p=\bra G\rho\ket G\ge M_G^{-1}\).  Then
\begin{equation}
 A_{\rm rel}(\rho)
 \ge
 p\log(pM_G)+(1-p)\log\!\left(
 \frac{1-p}{1-M_G^{-1}}
 \right).
 \label{eq:relative-labeling-cost}
\end{equation}
For the pure labeled state \(\rho=\ketbra GG\) with $M_G>1$, equality gives
\begin{equation}
 A_{\rm rel}(\ketbra GG)=\log M_G.
\end{equation}
\end{corollary}

\begin{proof}
Apply data processing for relative entropy to the binary measurement
\(\{\ketbra GG,\id-\ketbra GG\}\).  Its second-state success probability is
at most \(M_G^{-1}\).  For fixed \(p\ge M_G^{-1}\), binary relative entropy
is decreasing in the reference probability on \((0,p]\), which gives
\eqref{eq:relative-labeling-cost}.  For \(\ketbra GG\), the twirl is the
uniform mixture on \(\mathcal O_G\), and the stated equality is immediate.
\end{proof}

When $M_G=1$, the symmetry orbit is a singleton.  A pure labeled graph
state is already invariant and has $A_{\rm rel}=0$; the binary expression
in \eqref{eq:relative-labeling-cost} is then unnecessary.

Thus an invariant orbit superposition may represent an unlabeled isomorphism
class without selecting a preferred labeling, whereas interpreting a state as
one already selected labeled geometry requires a symmetry resource in the
state, measurement, environment, or boundary condition.  For an asymmetric
graph, \(M_G=N!\), so the pure-state relative-entropy cost is \(\log N!\).

\section{Verification and adversarial tests}
\label{sec:verification}

The ancillary program \texttt{verify\_relational\_histories.py} reproduces
finite-matrix and combinatorial tests independently of the LaTeX proofs.
It enumerates all labeled simple graphs for $2\le N\le5$ and verifies
the orbit--stabilizer identity and the singleton-orbit classification.
For $(N,L)=(2,4),(2,6),(3,4)$ it checks the update involutions,
commutation relations, gauge identity, the binomial data-penalty spectrum,
exact history-state kernel, and the full cyclic eigenvalue multiset and gap.
It also tests permutation covariance for $N=3$.

For the Boolean construction, the code checks the CNOT equivalence,
local eigenvectors and eigenvalues, additive block spectrum and fixed-weight
selection for $N=2,3$.  It independently evaluates the twirl and
labeling-concentration inequalities on seeded random density matrices.
The archived program and its machine-readable summary
\texttt{verification\_results.json} contain the precise test counts and
maximum residuals.  A separate script
\texttt{verify\_state\_concentration.py} supplies a smaller regression
for the state-level bound.  The figure can be regenerated by
\texttt{generate\_cyclic\_gap.py}.

Three failure modes delimit the finite checks: for a generic involutive
update, odd $L$ does not close the cyclic history; without the cyclic input
pin the zero-energy space has the full data dimension; and the first
pinned-block excitation lies below the first nonzero unpinned cycle mode.
These are finite regression tests of the formulas and tensor bookkeeping,
not numerical evidence of an emergent geometric phase or substitutes for
analytic proof.

\section{Discussion}

The prelocality obstructions separate three claims that are frequently
conflated.  Exact microscopic symmetry determines interaction orbits; it does
not determine a sparse metric.  Bounded-radius statistics determine a local
profile; they do not determine global connectedness.  Closed unitarity
preserves distinguishability; it does not produce norm attraction.  Each
missing implication requires a new mechanism rather than a change of wording.

The cyclic model is now spectrally complete.  Its $N$-independent gap is not
merely bounded by a Poincar\'e estimate: it is the lowest eigenvalue of a
unit-strength rank-one defect on a normalized cycle.  The defect equation
\eqref{eq:theta-equation} shows why the asymptotic constant is $\pi^2/2$ rather
than the first unpinned cycle constant $2\pi^2$.  The data volume changes only
multiplicities, not energies.  This exact separation is what makes the
certification bounds in \cref{thm:certification,thm:stability} transparent.

The Boolean construction shows that order-free histories need not inherit the
$L^{-2}$ gap of a single diffusive clock.  Independent local clocks convert the
clock graph into a product structure, and the local input condition changes
each pair block from a projector to the exactly solvable positive operator
\eqref{eq:local-boolean-term}.  The constant gap is purchased with $M$ clock
qubits and a control-sector ground manifold.  It is therefore a resource
tradeoff, not a route to universal constant-gap Hamiltonian computation.

A natural continuation is to couple the Boolean history manifold to a genuine
sparse-sector selector and then prove a propagation bound inside that sector.
The orbit-growth theorem implies that such a selector cannot be silently
identified with microscopic locality.  It must either carry a global witness,
use a thermodynamic mechanism with quantitative defect suppression, or create
an effectively local algebra after projection.  Any of these would be a new
phase theorem beyond the finite operator results proved here.

\section{Conclusion}

Exact permutation covariance before adjacency selection is prelocal.  Complete
$k$-body orbits have incident norm
$\binom{N-1}{k-1}|J_N|\norm h$ and cannot be finite range in a metric of
diverging diameter.  Bounded-radius graph profiles do not determine
connectedness, and unitary conjugation cannot contract distinct states to one
norm attractor.

Against this obstruction background, two relational history Hamiltonians are
solved exactly.  The cyclic model has a unique ground state, a complete
rank-one-defect spectral decomposition, and the exact gap
\eqref{eq:exact-gap}, with quantitative low-energy, Gibbs, and perturbative
certificates.  The Boolean model removes any ordering of commuting pair
updates, uses only four-local pair terms, has a completely classified ground
manifold, and retains the constant gap $E_*(1-2^{-1/2})$.  The comparison
isolates the real finite resources: compact clock versus local update
resolution, uniqueness versus sector multiplicity, and diffusive versus
constant gap.  Emergent geometry remains a separate theorem obligation.

\begin{acknowledgments}
The author thanks the developers of the open-source numerical and typesetting
tools used for independent finite-matrix checks and reproducible manuscript
construction.
\end{acknowledgments}

\bibliography{references}

\end{document}